\documentclass[12pt]{iopart}
\usepackage[ascii]{inputenc}
\usepackage[T1]{fontenc}
\usepackage[USenglish]{babel}
\usepackage{setstack}
\usepackage{iopams}
\usepackage{color}
\usepackage[usenames,dvipsnames]{xcolor}
\usepackage{amsthm}
\usepackage{amsopn}
\usepackage{graphicx}
\usepackage{esint}
\usepackage{float}
\usepackage{units}
\usepackage[all]{xy}
\makeatletter

\newtheorem{thm}{Theorem}
\newtheorem{defn}{Definition}
\newtheorem{cor}{Corollary}
\newtheorem{rem}{Remark}
\newtheorem{fact}{Fact}
\newtheorem{lema}{Lemma}
\newtheorem{prop}{Proposition}

\DeclareMathOperator{\diag}{diag}

\DeclareMathOperator{\vol}{vol}
\DeclareMathOperator{\Ad}{Ad}
\DeclareMathOperator{\Ker}{Ker}
\DeclareMathOperator{\Span}{Span}
\DeclareMathOperator{\GHZ}{GHZ}
\DeclareMathOperator{\SEP}{SEP}
\usepackage{color}

\makeatother

\begin{document}
\global\long\global\long\global\long\def\bra#1{\mbox{\ensuremath{\langle#1|}}}
\global\long\global\long\global\long\def\ket#1{\mbox{\ensuremath{|#1\rangle}}}
\global\long\global\long\global\long\def\bk#1#2{\mbox{\ensuremath{\ensuremath{\langle#1|#2\rangle}}}}
\global\long\global\long\global\long\def\kb#1#2{\mbox{\ensuremath{\ensuremath{\ensuremath{|#1\rangle\!\langle#2|}}}}}

\title
[Pure states of three qubits and their single-particle reduced density matrices]
{When is a pure state of three qubits determined by its single-particle reduced density matrices?}

\author{A Sawicki$^{1,2}$, M Walter$^{3}$, and M Ku\'{s}$^{2}$ }

\address{$^1$School of Mathematics, University of Bristol,
University Walk, Bristol BS8 1TW, UK}

\address{$^2$Center for Theoretical Physics, Polish Academy of Sciences, Al.
Lotnik\'ow 32/46, 02-668 Warszawa, Poland}

\address{{$^{3}$Institute for Theoretical Physics, ETH Z\"urich,
Wolfgang--Pauli--Str. 27, 8093 Z\"urich, Switzerland}}

\eads{\mailto{Adam.Sawicki@bristol.ac.uk\\mwalter@itp.phys.ethz.ch\\marek.kus@cft.edu.pl}}

\begin{abstract}
{Using techniques from symplectic geometry, we prove that a pure state of
three qubits is up to local unitaries uniquely determined by its one-particle
reduced density matrices exactly when their ordered spectra belong to the
boundary of the, so called, Kirwan polytope. Otherwise, the states with given
reduced density matrices are parameterized, up to local unitary equivalence,
by two real variables. Given inevitable experimental imprecisions, this means
that already for three qubits a pure quantum state can never be reconstructed
from single-particle tomography. We moreover show that knowledge of the
reduced density matrices is always sufficient if one is given the additional
promise that the quantum state is not convertible to the
Greenberger--Horne--Zeilinger (GHZ) state by stochastic local operations and
classical communication (SLOCC), and discuss generalizations of our results
to an arbitary number of qubits.}
\end{abstract}
\submitto{\JPA}
\maketitle

Finding solutions to Hamilton's equations for a given system is a standard
problem in classical mechanics. The system in question is often invariant
with respect to a certain group of symmetries, and when the symmetries are
continuous, such as in the case of rotational symmetry, they form a Lie group
$K$. The existence of symmetries is inevitably connected to the first
integrals of Hamilton's equations. The celebrated theorem of Arnold
\cite{Arnold89} states that in case when $K=T^{n}$ is an $n$-dimensional
torus and $n$ half the dimension of the
phase space $M$%
\footnote{A phase space is a symplectic manifold $(M,\omega)$, where $\omega$
is a closed nondegenerate two-form.} then the system is completely
integrable. When the group $K$ is not abelian, it is still possible to find
corresponding first integrals, although they typically do not Poisson
commute. For example, when a particle is moving in a potential with
rotational symmetry, so that $K=SO(3)$, then the conserved quantities are the
three components of the angular momentum, corresponding to the invariance of
the system with respect to infinitesimal rotations about the axes $x$, $y$,
and $z$. These infinitesimal rotations generate the Lie algebra
$\mathfrak{so}(3)$ of $SO(3)$. There are many possible bases of
$\mathfrak{so}(3)$, each corresponding to a different choice of rotation axes
and each giving three first integrals. The mathematical object which encodes
information about the first integrals for all possible choices of generators
of $\mathfrak{so}(3)$ is an equivariant map $\mu:
M\rightarrow\mathfrak{so}(3)^{\ast}$ from the phase space $M$ to the space of
linear functionals on the Lie algebra $\mathfrak{so}(3)$. For every
infinitesimal symmetry $\xi\in\mathfrak{so}(3)$ one obtains a corresponding
first integral by the formula $\mu_{\xi}(x)=\langle\mu(x),\,\xi\rangle$,
where by $\langle\,\cdot,\cdot\rangle$ we denote the pairing between linear
functionals from $\mathfrak{k}^\ast$ and vectors in $\mathfrak{k}$.  This
idea can be generalized to arbitrary Lie groups $K$ and a corresponding map
$\mu:M\rightarrow\mathfrak{k}^{\ast}$ is called a \emph{momentum map}
\cite{GS90}.

Remarkably, momentum maps appear naturally not only in classical but also in
quantum mechanics. Indeed, the Hilbert space $\mathcal{H}$ on which a given
quantum-mechanical system is modeled can be seen as a phase space if we
identify vectors that differ by a global rescaling or a phase factor
$e^{i\phi}$. The set of (pure) quantum states is thus isomorphic to the
complex projective space $\mathbb{P}(\mathcal{H})$, which is well-known to be
a symplectic manifold. When the considered system consists of $N$ subsystems%
\footnote{We assume for simplicity that their Hilbert spaces are of the same
dimension.} then the space $\mathcal{H}$ has the additional structure of a
tensor product, namely $\mathcal{H}=\mathcal{H}_{1}^{\otimes N}$, where
$\mathcal{H}_{1}$ is the single-particle Hilbert space. The mathematical
properties of this structure manifest physically as entanglement, i.e.\ 
quantum correlations between subsystems. These correlations are invariant
with respect to the local unitary action, i.e.\ to the action of the Lie
group $K=SU(\mathcal{H}_{1})^{\times N}$ on $\mathcal{H}$ by the tensor
product. Since $K$ preserves the phase space structure of
$\mathbb{P}(\mathcal{H})$ one gets a momentum map $\mu:
\mathbb{P}(\mathcal{H})\rightarrow\mathfrak{k}^{\ast}$, and the image
$\mu([v])$ of a quantum state $[v]\in\mathbb{P}(\mathcal{H})$ is, up to some
unimportant shifting, the collection of its one-body reduced density matrices
(see Section~\ref{sec:fibers}).
The matrices $\rho_{i}$ can be diagonalized by the action of $K$ and their
eigenvalues can be ordered, for example decreasingly. The map which assigns
to the state $[v]$ the collection of the ordered spectra of its one-body
reduced density matrices will be denoted by $\Psi :
\,\mathbb{P}(\mathcal{H})\rightarrow\mathfrak{t}_+^*$. This definition can be
generalized to an arbitrary compact Lie group $K$, with $\mathfrak{t}_{+}^*$
a positive Weyl chamber in $\mathfrak t$
(see Section~\ref{sec:momentum} for mathematical details and
Section~\ref{sec:fibers} for the case of three qubits).

Crucial for the rest of the paper is a specific geometric feature of the image of a momentum map, namely its convexity.
In the early '80s, the convexity properties of $\mu(M)$ were
investigated by Atiyah \cite{Atiyah82} and Guillemin and Sternberg
\cite{GS82,GS84} in the case of the abelian group $K=T^{n}$. They proved that
$\mu(M)$ is a convex polytope; it is in fact the convex hull of the image of
the set of fixed points. For nonabelian $K$, this is typically no longer
true. However, as it was shown in \cite{GS82} for K\"ahler manifolds and in
\cite{Kirwan82,Kirwan84} for general symplectic manifolds, the set $\Psi(M)
= \mu(M) \cap \mathfrak t^*_+$ is always a convex polytope---the so called
\emph{Kirwan polytope}, or \emph{momentum polytope}. This fact translates
to the observation that the set of possible local eigenvalues, i.e.\ the
eigenvalues of the one-partial reduced density matrices of an arbitrary pure quantum state
is a convex polytope. Finding an explicit description of this polytope is in general a
difficult problem; it is an instance of the quantum marginal problem, or $N$-representability problem in
the case of fermions \cite{ruskai69,colemanyukalov00}. Following the
groundbreaking work of Klyachko \cite{klyachko98} on the famous Horn's
problem, which can be formulated as the problem of determining a Kirwan
polytope, Berenstein and Sjamaar \cite{Berenstein00} found a general solution
for the case where $M$ is a coadjoint orbit of a larger group. This was in
turn used to solve the one-body quantum marginal problem; namely, sets of
inequalities describing the Kirwan polytope $\Psi(\mathbb P(\mathcal H))$
were given \cite{klyachko04,daftuarhayden04,klyachko06}. Intriguingly, the
Kirwan polytope can under certain assumptions be also described in terms of
coinvariants or, equivalently, in terms of the representations that occur in
the ring of polynomial functions on $M$ \cite{Ness84,Brion87}. In the context
of the quantum marginal problem, this has been discovered by Christandl
et.~al.~\cite{christandlmitchison06,christandlharrowmitchison07} and Klyachko
\cite{klyachko04}. The interplay of these two complementary perspectives can
also be seen in \cite{ChristandlDoranKousidisWalter}, where an algorithm has
been given for the more general problem of computing the distribution of
eigenvalues of the reduced density matrices of a multipartite pure state in
$\mathbb P(\mathcal H)$ drawn at random according to the Haar measure; in
particular, this gives an alternative solution to the one-body quantum
marginal problem.

The above line of work can also be seen as one of the first successful
applications of momentum map geometry to the theory of entanglement (see also
\cite{klyachko08}). In \cite{Sawicki11} the importance of the momentum map to
entanglement was investigated from a different perspective. The authors
showed that restricting the map $\Psi$ to different local unitary orbits in
$\mathbb{P}(\mathcal{H})$ gives rise to a well defined purely geometric
measure of entanglement. They also pointed out that the properties of the
fibers of $\Psi$ are crucial to the solution of the local unitary equivalence
problem and gave an algorithm for checking it for three qubits
\cite{sawicki11a}. Geometrically, the problem of local unitary equivalence of
two states reduces to determining whether they both belong to the same orbit
of the local unitary group $K=SU(\mathcal{H}_1)^{\times N}$. From the
physical point of view, this corresponds to checking if one of the states can
be obtained from the other one by unitary quantum operations on the
subsystems---an important problem in quantum engineering of states aiming at
practical applications \cite{Kraus10,Kraus10a}.

In a subsequent paper, the authors analyzed the geometric structure of the
fibers of $\Psi$ for two distinguishable particles, two fermions and two
bosons \cite{HKS12}. In all these cases, the $K^\mathbb{C}$-action is spherical, i.e.\
the Borel subgroup $B\subset K^{\mathbb{C}}$ has an open orbit in
$\mathbb{P}(\mathcal{H})$. By Brion's theorem \cite{Brion87}, this implies
that every fiber of $\Psi$ is contained in a single $K$-orbit. That is, in
these cases every quantum state is up to local unitaries uniquely determined
by the spectra of its reduced density matrices. Moreover, each fiber of
$\Psi$ has the structure of a symmetric space \cite{HKS12}.

In situations involving larger numbers of particles, e.g., $N$-qubit systems
with $N>2$, the action of $K^\mathbb{C}$ is not spherical and Brion's theorem cannot be
applied. The identification of $K$-orbits, i.e.\ classes of states which can
be mutually transformed into each other by local unitary transformations,
generically requires more information than it is contained in the spectra of
the single-particle reduced density matrices.

The present paper explores such situations where a simple application of Brion's theorem is no longer possible.
Its main goal is to analyze the set of quantum states which are mapped by $\Psi$ to the same
point of the Kirwan polytope $\Psi(M)$ (again, these are quantum states whose
reduced density matrices have the same spectra) by studying the fibers of
$\Psi$ or, equivalently, the symplectic quotients $M_\alpha =
\Psi^{-1}(\alpha)/K$, in the case where the action of $K^\mathbb{C}$ is no longer
spherical. Specifically, we consider the above-mentioned case of $N$ qubits,
where the Kirwan polytope is known explicitly \cite{Higuchi03,bravyi04}. A
detailed analysis is carried out for three qubits and we make some remarks
about the general case. Our main tools are the convexity theorem for
projective subvarieties \cite{Ness84,Brion87} as well as general properties of orbit spaces.
Specifically, we show that $M_\alpha$ is generically a two-dimensional
stratified symplectic space \cite{lermansjamaar91}. For the points in the
boundary of the Kirwan polytope the situation is very different. We prove
that in this case $M_\alpha$ is a single point, i.e.\ the dimension drops down
by two compared with the interior. In particular, these results characterize
the $K$-orbits which are uniquely determined by the spectra of the reduced
density matrices. They may be contrasted with the well-known fact that the
two-particle reduced density matrices generically suffice to determine a
pure-state of three qubits \cite{Linden02}.

The complexified group $G=K^\mathbb{C}$ plays its own role in the
classification of states of multiparticle quantum systems. Its elements
correspond to stochastic local operations with classical communication
(SLOCC); see e.g.\ \cite{horodecki09}. States can be classified by
identifying those which belong to the same SLOCC class, i.e.\ the same
$G$-orbit. As in the case of locally unitary equivalent states (orbits of
$K$), a detailed analysis of the corresponding Kirwan polytopes is useful in
such a classification. We therefore examine the Kirwan polytopes
$\Psi(\overline{G.[v]})$ for the three-qubit SLOCC classes \cite{Dur00},
i.e.\ for the orbit closures of the complexified group $G=K^{\mathbb{C}}$,
and describe their mutual relations. In particular, we find that the map
$\Psi$ separates $K$-orbits when restricted to the closure of the so-called
$W$-class. That is, states from the $W$-class are up to local unitaries
characterized by the collection of spectra of their one-qubit reduced density
matrices. Our argument generalizes directly to $W$-states of an arbitrary
number of qubits.

The paper is organized as follows. In Section~\ref{sec:momentum} we recall the
notion of a momentum map and state precisely various versions of the
convexity theorems. In Section~\ref{sec:fibers} we discuss in detail the
structure of the fibers of $\Psi$ for three qubits. Throughout the paper we
prove only new results and otherwise give references to the literature.

\section{Momentum maps} \label{sec:momentum}
Let $K$ be a compact connected Lie group acting on a symplectic manifold
$(M,\omega)$ by symplectomorphisms, i.e., the action $\Phi_{g} \colon M \rightarrow M$ of
any group element $g \in K$ preserves the symplectic form, 
$\Phi_{g}^{\ast}\omega=\omega$. Denote by $\mathcal{F}(M)$ the
space of smooth functions on $M$.

\begin{defn}
\label{moment_map_def}
We say that a symplectic action of $K$ on $(M,\omega)$ is Hamiltonian if and
only if there exists a momentum map $\mu:M\rightarrow\mathfrak{k}^{\ast}$, i.e.\ 
a map which satisfies the following three conditions:
\begin{enumerate}
\item For any $\xi\in\mathfrak{k}$, the fundamental vector field
    $\hat{\xi}(x)=\frac{d}{dt}\big|_{t=0}\Phi(e^{t\xi},\, x)$ is the
    Hamiltonian vector field for the Hamilton function
    $\mu_{\xi}(x)=\langle\mu(x),\,\xi\rangle$; i.e.
    $d\mu_{\xi}=\omega(\hat{\xi},\cdot)$.
\item The induced map $\mathfrak k \ni \xi \mapsto \mu_\xi \in \mathcal
    F(M)$ is a homomorphism of Lie algebras, i.e.
\begin{eqnarray*}
\mu_{[\xi_{1},\xi_{2}]}(x)=\{\mu_{\xi_{1}},\mu_{\xi_{2}}\}(x)\,.
\end{eqnarray*}

\item The map $\mu$ is equivariant, i.e.
    $\mu(\Phi_{g}(x))=\Ad_{g}^{\ast}\mu(x)$, where $\Ad_{g}^{\ast}$ is
    the coadjoint action of $K$ on $\mathfrak k^*$ defined by
    $\langle\Ad_{g}^{\ast}\alpha,\,\xi\rangle=\langle\alpha,\Ad_{g^{-1}}\xi\rangle$
    in terms of the adjoint action $\Ad_g\xi=\frac{d}{dt}|_{t=0}g\,e^{t\xi}\,g^{-1}$ of $K$ on its Lie algebra $\mathfrak{k}$.
\end{enumerate}
\end{defn}

\noindent
For semisimple $K$, hence in particular for $K=SU(\mathcal{H}_{1})^{\times N}$,
the momentum map $\mu$ is uniquely defined by the above properties \cite{GS90}.

\subsection{Convexity properties of the momentum map}
\label{subs:convexity} We will now assume that $M$ is compact and connected.
Let us choose a maximal torus $T\subset K$, with Lie algebra $\mathfrak{t}$,
and a positive Weyl chamber $\mathfrak t^*_+ \subset \mathfrak t^*$. Denote
by $\Psi:M\rightarrow\mu(M)\cap\mathfrak{t}^*_{+}$ the map which assigns to
$x\in M$ the unique point of intersection $\mu(K.x)\cap\mathfrak{t}_{+}^*$.
Then the following convexity results hold:
\begin{enumerate}
\item The image $\Psi(M)$ is a convex polytope, the so-called
    \emph{Kirwan polytope} (\cite{GS82} and \cite{Kirwan82,Kirwan84}).
\item The fibers of $\mu$ (and hence the fibers of $\Psi$) are connected
    \cite{Kirwan82,Kirwan84}.
\item The map $\Psi$ is an open map onto its image, i.e.\ for any open
    subset $U\subset M$ the image $\Psi(U)$ is open in $\Psi(M) = \mu(M)
    \cap \mathfrak t^*_+$ \cite{Knop02}.
\item The image $\Psi(M_{\max})$ of $M_{\max}=\{x\in M:\,\dim K.x\,\,\text{is maximal}\}$ is convex \cite{HH96}.
\end{enumerate}

\noindent Since $M_{\max}$ is a connected, open and dense subset of $M$ \cite{Montgomery59}, (iv)
implies the following:
\begin{fact}\label{interior}
The set $\mu(M_{\max})\cap\mathfrak{t}_{+}^*$ is an open, dense, connected,
convex subset of $\Psi(M) = \mu(M)\cap\mathfrak{t}_{+}^*$. In particular, it
contains the (relative) interior of the Kirwan polytope.
\end{fact}

\noindent Denote by $G = K^{\mathbb C}$ the complexification of $K$. This is
a complex reductive group.
Let us assume as in \cite{Brion87} that $M \subset \mathbb P(\mathcal H)$ is
a non-singular $G$-invariant irreducible subvariety of the complex projective
space associated with a rational $G$-representation $\mathcal H$, and let us
choose a $K$-invariant inner product on $\mathcal H$. Then $M$ is a
symplectic manifold when equipped with the restriction of the Fubini--Study
form,
\begin{equation}
  \label{eq:fubinistudy}
  \omega_{[v]}(\hat{A}_{[v]}, \hat{B}_{[v]})
  = 2\,\mathrm{Im} \frac{\bk{A v}{B v}}{\bk vv}
  = - i \frac{\bk{[A,B]v}v}{\bk vv}
  \qquad
  \forall A,B\in\mathfrak{u}(\mathcal{H}),
\end{equation}
where $[v] \in \mathbb P(\mathcal H)$ is the projection of a vector
$v\in\mathcal{H}$, and where $\hat{A}$ is the fundamental vector field
generated by the action of $A \in \mathfrak u(\mathcal H)$. More concretely,
we can represent the tangent space $T_{[v]}\mathbb{P}(\mathcal{H})$ by
$\mathcal H / \mathbb C v \cong (\mathbb C v)^\perp$. In this picture, the
tangent vector $A_v$ is given by the orthogonal projection of $A v$ onto
$(\mathbb C v)^\perp$. $M$ also carries a canonical momentum map $\mu : M
\subset \mathbb{P}(\mathcal{H})\rightarrow\mathfrak{k}^{\ast}$
\begin{equation}
  \label{momentPH}
  \langle\mu([v]),\, A\rangle = - i \frac{\bk v{Av}}{\bk vv} \quad
\forall A\in\mathfrak{k},
\end{equation}
In this situation, there are further convexity results:

\begin{enumerate}
\item[(v)] The image $\Psi(\overline{G.x})$ is a convex polytope
    \cite{Brion87,Ness84}.
\item[(vi)] The collection of different polytopes $\Psi(\overline{G.x})$,
    where $x$ ranges over $M$, is finite \cite{GuliSjamaar06}.
\end{enumerate}

\subsection{Fibers of the momentum map}

The following two facts characterizing kernel and image of the differential
of the momentum map are well-known consequences of the definition
\cite{GS82}:
\begin{fact} 
\label{kernel_mu} The kernel of $d\mu\big|_x :
T_{x}M\rightarrow\mathfrak{k}^*$ is equal to the $\omega$-orthogonal
complement of $T_{x}(K.x)$, i.e.
\begin{equation*}
  \Ker(d\mu\big|_x) =
  \{ Y \in T_x M : \omega(\hat\xi, Y) = 0 \quad \forall \xi \in \mathfrak k \} =
  \left( T_{x}(K.x) \right)^{\perp\omega}.
\end{equation*}
\end{fact}

\begin{fact} 
\label{range_mu} The image of $d\mu\big|_x : T_{x}M\rightarrow\mathfrak{k}^*$
is equal to the annihilator of $\mathfrak{k}_{x}$, the Lie algebra of the
isotropy subgroup $K_{x} \subset K$.
\end{fact}

\noindent These results have immediate consequences for the characterization
of the fibers of $\mu$:
Let us assume that $d\mu$ is surjective at a single point $x \in M$ (as will be the
case in our applications).
Using Fact~\ref{range_mu}, we observe that $d\mu$ is
surjective if and only if the Lie algebra $\mathfrak k_{x}$ is trivial. This
in turn happens if and only if the isotropy subgroup $K_{x}$ is discrete.
Therefore, $d\mu$ is automatically surjective at all points in $M_{\max}$.
In particular, the implicit function theorem together with Fact~\ref{kernel_mu}
implies that for all $x\in M_{\max}$, the tangent space at $x$ of the
$\mu$-fiber through $x$ has dimension
\begin{eqnarray}
  \label{Mmindim}
  \dim T_{x} (\mu^{-1}(\mu(x))) =
  \dim \left( T_{x}(K.x) \right)^{\perp\omega} =
  \dim M - \dim \mathfrak k^*,
  \label{eq:regular}
\end{eqnarray}
and hence:

\begin{fact}
\label{Mmin} For the points $x\in M_{\max}$, i.e.\ for an open, connected and
dense subset of $M$, the intersection $\mu^{-1}(\mu(x))\cap M_{\max}$ of the
$\mu$-fiber through $x$ with $M_{\max}$ is a $\dim \left( T_x (K.x)
\right)^{\perp\omega}$-dimensional manifold.
\end{fact}


\noindent Understanding the fibers $\mu^{-1}(\alpha)$ away from the regular
points in $M_{\max}$ is a more delicate problem. Since each fiber
$\mu^{-1}(\alpha)$ is $K_\alpha$-invariant, where $K_\alpha$ is the isotropy
subgroup of $\alpha$ with respect to coadjoint action, it is convenient to
introduce the symplectic quotient $M_\alpha := \mu^{-1}(\alpha) / K_\alpha =
\Psi^{-1}(\alpha) / K$. It is in general a stratified symplectic space
\cite{lermansjamaar91}, which is moreover connected by property (ii) in
Subsection~\ref{subs:convexity}. The following is then immediate:

\begin{fact}
  The fiber $\mu^{-1}(\alpha)$ intersects a single $K$-orbit if and only if
  $M_\alpha$ is a single point or, equivalently, if and only if $M_\alpha$ is
  zero-dimensional.
\end{fact}

If $\alpha \in \Psi(M_{\max})$ then the maximal-dimensional stratum of
$M_\alpha$ is simply $(\mu^{-1}(\alpha) \cap M_{\max}) / K_\alpha$
\cite{meinrenken97}. Since the isotropy subgroup of any point in $M_{\max}$
is discrete,
\begin{equation}
  \label{symplquotdim}
  \dim M_\alpha = \dim M - \dim \mathfrak k^* - \dim K_\alpha
\end{equation}
In other words, we can by mere dimension counting determine whether a
$K$-orbit in $M_{\max}$ is uniquely determined by its image under the
momentum map.

\section{Fibers of the momentum map for three qubits}
\label{sec:fibers}

With the momentum map machinery presented in the preceding section we are now
well-equipped to analyze when a pure state of three qubits is, up to local
unitaries, determined by the spectra of its reduced density matrices.

To this end, let $M=\mathbb{P}(\mathcal{H})$ be the projective space of pure
states associated with the three-qubit Hilbert space
$\mathcal{H}=\mathbb{C}^{2}\otimes\mathbb{C}^{2}\otimes\mathbb{C}^{2}$. The
group of local unitaries $K=SU(2)^{\times3}$ and its complexification
$G=SL(2)^{\times3}$ act on $M$ by the tensor product. The Lie algebra of $K$
is
$\mathfrak{k}=\mathfrak{su}(2)\oplus\mathfrak{su}(2)\oplus\mathfrak{su}(2)$,
i.e.\ triples of traceless antihermitian matrices. As described in the
preceding section, $M$ is a symplectic manifold with respect to the
Fubini--Study form (\ref{eq:fubinistudy}) and the $K$-action is Hamiltonian
with canonical momentum map (\ref{momentPH}).

It is easy to see that under the identification of $\mathfrak{k}^*$ with
$\mathfrak{k}$ induced by the trace inner product, the image $\mu([v])$ is
given by the collection of one-qubit reduced density matrices, namely
\begin{equation}
  \label{momentum map three qubits}
  \mu([v]) = i (\rho_{1}-\frac{1}{2}I,\,\rho_{2}-\frac{1}{2}I,\,
\rho_{3}-\frac{1}{2}I),
\end{equation}
where $I$ is the $2\times2$ identity matrix (see, e.g.,
\cite{Sawicki11,ChristandlDoranKousidisWalter}).

Let us fix the maximal torus $T \subset K$ to be the set of unitary diagonal
matrices with determinant equal to one. Then the Lie algebra $\mathfrak{t}$ is
equal to the space of traceless antihermitian diagonal matrices.
We choose as the positive Weyl chamber the following set of matrices:
\begin{equation}
  \label{positive weyl chamber}
  \mathfrak{t}^*_{+} =
  \left\{
  (\diag(-i\lambda_{1},i\lambda_{1}),\,\diag(-i\lambda_{2},i\lambda_{2}),\,
\diag(-i\lambda_{3},i\lambda_{3}))
  :
  \lambda_i \geq 0
  \right\}.
\end{equation}
It follows that, up to some rescaling and shifting, the map $\Psi$ sends a
pure state $[v]$ to the (ordered) spectra of its reduced density matrices
$\rho_1, \rho_2, \rho_3$.

The following theorem describes the Kirwan polytope in terms of inequalities
for a general system of $N$ qubits:

\begin{fact}[\cite{Higuchi03}]
\label{higuchi}
For a $N$-qubit system, the constraints on the one-qubit reduced density
matrices $\rho_i$ of a pure state are given by the \emph{polygonal inequalities}
\begin{eqnarray*}
  p_{i}\leq\sum_{j\neq i}p_{j},
\end{eqnarray*}
where $p_{i} \leq \frac 1 2$ denotes the minimal eigenvalue of $\rho_{i}$
($i=1,\ldots,N$).
\end{fact}

\noindent The Kirwan polytope for three qubits is shown in
Figure~\ref{figure:kirwan-three-qubits}.

\begin{figure}[H]
~~~~~~~~~~~~~~~~~~~~~~\includegraphics[scale=0.6]{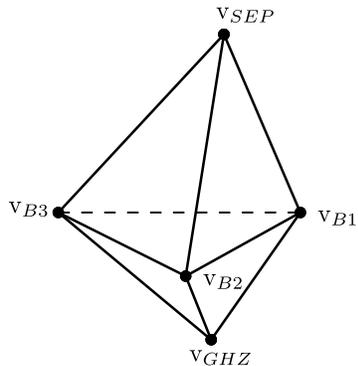}

\caption{The Kirwan polytope $\Psi(M)=\mu(M)\cap\mathfrak{t}_{+}^{\ast}$ for
three qubits.}
\label{figure:kirwan-three-qubits}
\end{figure}

\noindent This polytope has $5$ vertices, $9$ edges and $6$ faces.
As a convex set it is of course generated by its vertices, which are
\begin{eqnarray}
\label{eq:vertices}
\mathrm{v}_{\SEP}=i\{\diag(-\nicefrac{1}{2},\nicefrac{1}{2}),\,\diag(-\nicefrac{1}{2},\nicefrac{1}{2}),\,\diag(-\nicefrac{1}{2},\nicefrac{1}{2})\}, \nonumber \\
\mathrm{v}_{B1}=i\{\mathrm{\diag\ensuremath{(-\nicefrac{1}{2},\nicefrac{1}{2}),\,}diag}(0,0),\,\diag(0,0)\}, \nonumber\\
\mathrm{v}_{B2}=i\{\diag(0,0),\,\diag(-\nicefrac{1}{2},\nicefrac{1}{2}),\mathrm{,\,\ diag}(0,0)\},\\
\mathrm{v}_{B3}=i\{\diag(0,0),\,\diag(0,0),\,\diag(-\nicefrac{1}{2},\nicefrac{1}{2})\}, \nonumber\\
\mathrm{v}_{\GHZ}=i\{\mathrm{\diag\ensuremath{(0,0),\,}diag}(0,0),\,\diag(0,0)\}.\nonumber
\end{eqnarray}

\noindent Since the polytope is full-dimensional, Sard's theorem
\cite{sternberg64} implies the existence of a regular point in $M$, i.e.\ a
point with discrete isotropy subgroup (see also \cite{CS00}). Hence,
\begin{lema}
  The set $M_{\max} \subset M$ is connected, open and dense and consists of
  orbits of dimension $\dim K = 9$.
\end{lema}

\noindent We will now analyze points inside the interior of the Kirwan
polytope. Notice first that by Fact~\ref{interior} the preimage of any
such point $\alpha$ contains a point $x\in M_{\max}$. Therefore,
Fact~\ref{Mmin} and (\ref{Mmindim}) show that $\mu^{-1}(\alpha) \cap
M_{\max}$ is a manifold of dimension
\begin{equation*}
  \dim (\mu^{-1}(\alpha) \cap M_{\max}) =
  \dim (T_{x}K.x)^{\perp\omega} =
  \dim M - \dim \mathfrak k^* = 14 - 9 = 5.
\end{equation*}
Since $K_\alpha = T$ for points in the interior of the positive Weyl chamber,
this manifold consists of $3$-dimensional $T$-orbits, and
(\ref{symplquotdim}) implies that
\begin{equation*}
  \dim M_\alpha = 5 - 3 = 2.
\end{equation*}
If we replace $\mu^{-1}(x)$ by $\Psi^{-1}(x) = K . \mu^{-1}(x)$, we replace
each $T$-orbit by a $K$-orbit and therefore increase the dimension by
\begin{equation*}
  \dim \Omega_\alpha = \dim \nicefrac K {K_\alpha} = \dim K - \dim T = 6,
\end{equation*}
the dimension of the corresponding coadjoint orbit $\Omega_\alpha = K .
\alpha \subset \mathfrak k^*$.


Summing up, we proved:

\begin{thm}\label{interior-char}
  For any point $\alpha$ inside the interior of the Kirwan polytope there
  exists a point $x\in M_{\max}$ such that $\alpha=\mu(x)$. The manifold
  $\mu^{-1}(\alpha) \cap M_{\max}$ is $5$-dimensional, consisting of
  $3$-dimensional $T$-orbits. Moreover, $\Psi^{-1}(\alpha) \cap M_{\max}$ is
  an $11$-dimensional manifold consisting of $9$-dimensional orbits $K.y$
  with the property $\mu(K.y)=\Omega_\alpha$, and the symplectic quotient
  $M_\alpha$ has dimension $2$.
\end{thm}

\noindent The following is a direct consequence of Theorem \ref{interior-char}
(cf.\ the discussion at the end of Section~\ref{sec:momentum}):

\begin{cor}
The orbits $K.x$ for which $\mu(K.x)\cap\mathfrak{t}_{+}^*$ belongs to the
interior of Kirwan polytope \emph{cannot} be separated by the momentum map.
In other words, a pure state of three qubits whose spectrum is non-degenerate
and satisfies the polygonal inequalities with strict inequality is
\emph{never} determined up to local unitaries by the spectra of its reduced
density matrices.
\end{cor}

\noindent What is left is to analyze points in the boundary of the Kirwan
polytope. We postpone this to the end of the section (see Subsection
\ref{sec:boundary}).

\subsection{The SLOCC classes and their Kirwan polytopes}

As mentioned in the introduction, the classification of $G$-orbits, where
$G=K^\mathbb{C}$ gives another view on the entanglement properties of states
on the composite system. In this subsection we explicitly compute the Kirwan
polytopes $\Psi(\overline{G.x})$ for all $G$-orbit closures, and show how
they are related to the polytope $\Psi(\mathbb{P}(\mathcal{H}))$. Physically,
$G$-orbits correspond to SLOCC classes \cite{Dur00}, hence our results
describe the spectra of the reduced density matrices of the quantum states in
each SLOCC entanglement class. The problem of classifying $G$-orbits in
$\mathbb{P}(\mathcal{H})$ is inherently connected to the momentum map
geometry as well as to the construction of the so-called Mumford quotient
\cite{Kirwan84,Ness84}. We explain this connection in
\cite{SOK12,WalterDoranGrossChristandl}.

\paragraph{Classification of $G$-orbits}
It has been shown in \cite{Dur00} that there are six SLOCC entanglement
classes, i.e.\ $G$-orbits in $\mathbb{P}(\mathcal{H})$. For convenience of
the reader we list them below and briefly summarize their basic geometric
properties:
\begin{enumerate}
\item The $G$-orbit of the Greenberger--Horne--Zeilinger state, $x_{\GHZ}
    = [\frac{1}{\sqrt{2}} \left( \ket{000}+\ket{111} \right)]$
    \cite{GHZ}: It is the open dense orbit, hence of (real) dimension
    $\dim \mathbb{P}(\mathcal{H}) = 14$.
\item The $G$-orbit of the $W$-state, $x_W =[\frac{1}{\sqrt{3}} \left(
    \ket{100}+\ket{010}+\ket{001} \right)]$: Here, $\dim G . x_W = 12$,
    while $\dim K . x_W = 8$. For the proof, it is enough to compute the
    dimension of the tangent space $T_{x_W} (G.x_W)$, which can be
    represented as the projection of $\Span_{\mathbb C} \{A x_W : A \in
    \mathfrak g \}$ onto the orthogonal complement of $x_W$. The Lie
    algebra $\mathfrak{g}$ is equal to $\mathfrak{sl}_2(\mathbb C) \oplus
    \mathfrak{sl}_2(\mathbb C) \oplus \mathfrak{sl}_2(\mathbb C)$, where
\begin{eqnarray*}
  \mathfrak{sl}_2(\mathbb C) = \Span_{\mathbb C} \left\{ E_{12}, \, E_{21},
\, E_{11} - E_{22} \right\},
\end{eqnarray*}
and $E_{ij}$ is the $2\times2$ matrix with a single non-zero entry equal
to one in the $i$-th row and $j$-th column. It is easy to see that the
following seven vectors
\begin{eqnarray*}
  (E_{12}\otimes I\otimes I)x_W=(I\otimes E_{12}\otimes I)x_W=(
  I\otimes I\otimes E_{12})x_W \propto \ket{000},\\
  (E_{21}\otimes I\otimes I)x_W \propto \ket{110}+\ket{101}, \\
  (I\otimes E_{21}\otimes I)x_W \propto \ket{110}+\ket{011}, \\
  (I\otimes I\otimes E_{21})x_W \propto \ket{101}+\ket{011}, \\
  ((E_{11}-E_{22}) \otimes I \otimes I) x_W \propto -\ket{100} + \ket{010}
  + \ket{001}, \\
  (I \otimes (E_{11}-E_{22}) \otimes I) x_W \propto +\ket{100} - \ket{010}
  + \ket{001}, \\
  (I \otimes I \otimes (E_{11}-E_{22})) x_W \propto +\ket{100} + \ket{010}
  - \ket{001}
\end{eqnarray*}
span a complex vector space of dimension $6$ after projection onto
$x_W^\perp$ (the last $3$ vectors become linearly dependent). We conclude
that $\dim G.x_W=2 \cdot 6 = 12$. Similarly, one shows that $\dim
K.x_W=8$. For future reference we denote
\begin{equation}\label{w-vertex}
\mathrm{v}_W=\mu(x_W)=i\{\diag(-\nicefrac{1}{6},\nicefrac{1}{6}),\,
\diag(-\nicefrac{1}{6},\nicefrac{1}{6}),\,\diag(-\nicefrac{1}{6},
\nicefrac{1}{6})\}.
\end{equation}

\item The $G$-orbits through the bi-separable Bell states $x_{B1}=[\frac
    1 {\sqrt 2}(\ket{000}+\ket{011})]=[\ket 0 \otimes \frac 1 {\sqrt 2}
    \left(\ket{00} + \ket{11}\right)]$, $x_{B2}=[\frac 1 {\sqrt
    2}(\ket{000}+\ket{101})]$, and $x_{B3}=([\frac 1 {\sqrt
    2}(\ket{000}+\ket{110})]$: In a similar way as for $x_W$ one can show
    by an explicit computation that $\dim G.x_{Bk}=8$ and $\dim
    K.x_{Bk}=5$.
\item The $G$-orbit of separable states, generated by
    $x_{\SEP}=[\ket{000}]$: By the Kostant--Sternberg theorem, it is the
    unique symplectic orbit \cite{Sawicki11}
and $\dim G.x_{\SEP} = \dim K.x_{\SEP} = 9 - 3 = 6$.
\end{enumerate}

\paragraph{Kirwan polytopes of $G$-orbit closures}

Let us write $X_j := \overline {G \cdot x_j}$ for the closures of these
$G$-orbits. The representatives $x_j$ that we have chosen above satisfy the
following property:

\begin{lema}
  \label{The-distance} We have $\Psi(x_j) = \mu(x_j) = \mathrm{v}_j$, where
  $\mathrm{v}_j$ is defined in (\ref{eq:vertices}, \ref{w-vertex}). Moreover,
  $\Psi(x_j)$ is the closest point to the origin in the Kirwan polytope
  $\Psi(X_j)$.
\end{lema}

\noindent The first claim follows from a simple computation. The fact that
$\mathrm{v}_j$ is the closest point to the origin of the Kirwan polytope of
the corresponding $G$-orbit closure follows from the general theory of
\cite{Kirwan82}: the gradient descent with respect to the norm square of the
moment map is at any point $x$ implemented by the vector field generated by
$i \mu(x) \in i \mathfrak k \subseteq \mathfrak g$, and one easily checks
that $\widehat{i \mu(x_j)}\big|_{x_j} = 0$. It will be explained in more
detail in \cite{SOK12,WalterDoranGrossChristandl}.

We will now describe the Kirwan polytopes of the $G$-orbit closures: Since $G
. x_{\GHZ}$ is dense in $\mathbb P(\mathcal H)$, we immediately obtain that
$\Psi(X_{\GHZ})$ is equal to the full Kirwan polytope for three qubits as
described by Fact \ref{higuchi}. On the other hand, since $G . x_{\SEP} = K .
x_{\SEP}$ is a single $K$-orbit, $\Psi(X_{\SEP})$ is a single point, namely
$\Psi(X_{\SEP}) = \{\mathrm{v}_{\SEP}\}$. Similarly, one finds that the
Kirwan polytopes for the bi-separable Bell states $x_{Bk}$ are equal to the
one-dimensional line segments from $\mathrm{v}_{Bk}$ to $\mathrm{v}_{\SEP}$.
Therefore, the only non-trivial task is the computation of the Kirwan
polytope for the $W$-state:

\begin{prop}
  \label{kirwan w}
  The Kirwan polytope $\Psi(X_W)$ is equal to the convex hull of the points
  $\mathrm{v}_{B1}$, $\mathrm{v}_{B2}$, $\mathrm{v}_{B3}$ and $\mathrm{v}_{\SEP}$.
  In particular, it is of maximal dimension.
\end{prop}
\begin{proof}
  Clearly, $\Psi(X_W)$ is a subset of the full Kirwan polytope $\Psi(\mathbb
  P(\mathcal H))$. On the other hand, Lemma \ref{The-distance} and convexity
  imply that it is also contained in the half-space through $\mathrm{v}_W$
  with normal vector $\mathrm{v}_W$. Since the intersection of this
  half-space with $\Psi(\mathbb P(\mathcal H))$ is precisely equal to the
  convex hull of the points $\mathrm{v}_{B1}$, $\mathrm{v}_{B2}$,
  $\mathrm{v}_{B3}$ and $\mathrm{v}_{\SEP}$, we only need to show that these
  points are contained in the Kirwan polytope.

  We will in fact show that the corresponding preimages $x_j$ are contained
  in the orbit closure $X_W$. For this, we observe that the action of the
  complexification $T^{\mathbb{C}} \subset G$ of the maximal torus $T\subset
  K$ applied to $x_W$ gives rise to all states of the form
  \begin{equation}
    \label{some W class states}
    [c_1 \ket{100} + c_2 \ket{010} + c_3 \ket{001}]
    \quad
    (c_j \neq 0)
  \end{equation}
  In particular, $x_j \subseteq X_W = \overline{G \cdot x_W}$ for
  $j=B1, B2, B3, \SEP$, and the claim follows.
\end{proof}

\subsection{Sphericality of the $W$ SLOCC class}

In this subsection, we show that $X_W = \overline{G.x_W}$ is a spherical
variety. It then follows by Brion's theorem that every quantum state in $X_W$
is, up to local unitaries, characterized by the collection of the spectra of
its one-qubit reduced density matrices. In other words, $\Psi$ separates the
$K$-orbits in $X_W$.

We start by clarifying the geometric structure of $X_W$. Note that since
$G=K^{\mathbb{C}}$ is a complex reductive group, any orbit closure $X_j$ is a
$G$-invariant irreducible subvariety of $\mathbb{P}(\mathcal{H})$; however,
these varieties will in general be singular. This is in fact already the case
for $X_W$. Indeed, it is known from \cite{klyachko08} that
\begin{equation*}
  X_W = \mathbb P(\mathcal H) \setminus (G . x_{\GHZ}) =
  \{ [v] \in \mathbb P(\mathcal H) : \mathrm{Det}(v) = 0 \},
\end{equation*}
where $\mathrm{Det}$ is the Cayley hyperdeterminant (the basic invariant for
the $G$-representation $\mathcal H$), and one readily verifies that the
tangent space at $x_{\SEP} \in X_W$ has complex dimension $7 > 6 =
\dim_{\mathbb{C}} X_W$.

Let us denote by $\mu_W:X_W \rightarrow \mathfrak k^*$ the restriction of
$\mu:\mathbb{P}(\mathcal{H})\rightarrow\mathfrak{k}^{\ast}$ to $X_W$. Our aim
now is to prove that $\mu_W$ separates the $K$-orbits in $X_W$. To this end
we use the following theorem by Brion (see \cite{Brion87}, and also
\cite{huckleberry90}):

\begin{fact}
  \label{Brion}
  Let $G = K^{\mathbb C}$ be a connected complex reductive group,
  $\mathcal H$ a rational $G$-representation,
  and $X$ a $G$-invariant irreducible subvariety of $\mathbb{P}(\mathcal{H})$
  (cf.~Subsection \ref{subs:convexity}).
  Then the following are equivalent:
  \begin{enumerate}
  \item $X$ is spherical, i.e.\ the (every) Borel subgroup $B$ has a
      Zariski-open orbit in $X$.
  \item For every $x \in X$ the fiber $\mu^{-1}(\mu(x))$ is contained in
    a single $K$-orbit, $K . x$.
  \end{enumerate}
\end{fact}

We will now show that indeed:

\begin{prop}
  \label{borel}
  The $G$-variety $X_W$ is spherical.
\end{prop}
\begin{proof}
  Consider the Borel subgroup $B$, which consists of the lower-triangular
  matrices in $G$.
Its Lie algebra $\mathfrak b$ is equal to $\mathfrak t^{\mathbb C} \oplus
\mathfrak n$, where
  \begin{equation*}
    \mathfrak{n} = \Span_{\mathbb{C}} \{ E_{21} \otimes I\otimes I,\,
I\otimes E_{21}\otimes I,\, I\otimes I\otimes E_{21}\}.
  \end{equation*}
In order to show that $X_W$ is spherical, we have to show that there exists a
Zariski-open orbit $B.x$. Since $B.x$ is Zariski-open in its closure, it
suffices to show that the closure of $B.x$ is equal to $X_W$. Now, since the
closure of $B.x$ is a closed subvariety, it is either equal to $X_W$ or of
lower dimension. Therefore, it suffices to show that the dimension of $B.x$
is equal to the dimension of $X_W$. Since $B.x$ is a smooth variety, we can
compute its dimension by computing the dimension of the tangent spaces at any
point. Hence in order to show that $X_W$ is spherical, it suffices to show
that $\dim_{\mathbb C} T_x (B.x) = \dim_{\mathbb C} X_W = 6$ for any single
point $x$ in the $G$-orbit of $x_W$. We will consider the state
  \begin{equation*}
    x
    = [ \frac 1 {\sqrt 4} \left( \ket{100}+\ket{010}+\ket{001}
    +\ket{000} \right) ]
    = ( \left(\begin{array}{cc} 1 & 1 \\ 0 & 1 \end{array}\right) \otimes I \otimes I ) \cdot x_W.
\end{equation*}
Indeed, one can easily verify that the tangent vectors generated by
$\mathfrak b$, i.e.\ the projection of
\begin{eqnarray*}
(E_{21} \otimes I \otimes I) x  \propto  \ket{110} + \ket{101} + \ket{100}, \\
(I \otimes E_{21} \otimes I) x  \propto  \ket{110} + \ket{011} + \ket{010}, \\
(I \otimes I \otimes E_{21}) x  \propto  \ket{101} + \ket{011} + \ket{001}, \\
((E_{11} - E_{22}) \otimes I \otimes I) x  \propto  -\ket{100} + \ket{010} + \ket{001} + \ket{000}, \\
(I \otimes (E_{11} - E_{22}) \otimes I) x  \propto  +\ket{100} - \ket{010} + \ket{001} + \ket{000}, \\
(I \otimes I \otimes (E_{11} - E_{22})) x  \propto  +\ket{100} + \ket{010} - \ket{001} + \ket{000},
\end{eqnarray*}
onto the orthogonal complement $x^\perp$, span a complex vector space of
dimension six.
\end{proof}

\noindent The following is now a direct consequence of
Proposition~\ref{borel} and Fact~\ref{Brion}:

\begin{thm}\label{W sep}
  The momentum map $\mu_W$ separates all $K$-orbits inside $X_W$. That is,
  every quantum state in the $W$ SLOCC class is (up to local unitaries)
  uniquely determined by the spectra of its reduced density matrices.
\end{thm}

\begin{rem}
  In fact, any $W$-type state of $L$ qubits, i.e.\ any state which is in the
  SLOCC class of
\begin{equation*}
 x_W = \left[ \frac{1}{\sqrt{N}} \left( \ket{10\ldots 0} + \ket{010\ldots 0}
+ \ldots + \ket{0\ldots 01} \right) \right] \in
\mathbb{P}((\mathbb{C}^2)^{\otimes L}),
\end{equation*}
is up to local unitary equivalence determined by its single-particle density
matrices. This can be shown \emph{mutatis mutandis} as for the case of three
qubits (cf.~\cite{Yu12} for an alternative proof based on linear-algebraic
computations, which appeared after the initial version of this paper). By
Brion's theorem, it is enough to show that the $G$-variety $\overline{G.x_W}$
is spherical. To this end, we choose the state
  \begin{equation*}
    x= \left(\left(\begin{array}{cc}
    1 & 1\\
    0 & 1
    \end{array}\right)\otimes I\ldots\otimes I\right)x_W=
\left[ \frac{1}{\sqrt{N+1}}\left(\ket{0\ldots0}+\sqrt{N}x_W\right) \right]
\in G . x_W
\end{equation*}
in complete analogy to the proof of Proposition~\ref{borel}. It is again a
matter of simple calculations to show that the real dimensions of the orbits
$G.x$ and $B.x$, where $B$ is a Borel subgroup defined as in the case of
three qubits, are equal (to $4L$). Hence $\overline{G.x_W}$ is a spherical
variety.
\end{rem}

\begin{rem}
{The fact that $W$-states are, up to local unitaries, uniquely determined by
the spectra of their reduced one-particle density matrices can be, in
principle, also determined in purely algebraic manner from the explicit
knowledge of invariants separating different orbits of the group of local
unitary transformations.
Indeed, if we are able to show that for all states of the $W$-class all
invariants depend only on the spectra of one-particle density matrices then
the conclusion follows. For the three-qubit case this can be, in principle,
done using results of \cite{Acin00,Acin01}, where explicit forms of the
invariants are given. For $W$-states one can, after some calculations, prove
that they depend only on one-particle spectra. For more than three qubits,
however, this approach quickly becomes intractable: one first needs to
explicitly compute all relevant invariants (which has to be done separately
for each number of qubits) and then prove that for states of the $W$-type the
result depends only on the local spectra. In contrast, the geometric approach
presented above leads immediately to the desired result for an arbitrary
number of qubits.}
\end{rem}


Combining Proposition~\ref{borel} with Theorem~\ref{interior-char} we
conclude that every manifold $\mu^{-1}(\Omega_{\alpha})\cap M_{\max}$, where $\alpha$ is
in the interior of the polytope $\Psi(X_W)$ contains a unique orbit $K.\tilde{x}$ with $\tilde{x}\in G.x_W$. We have
illustrated the situation in Figure~\ref{figure:slices}.

\begin{figure}[h]
~~~~~~~~~~~~~~~~~~~~~~~~~~~~~\includegraphics[scale=0.7]{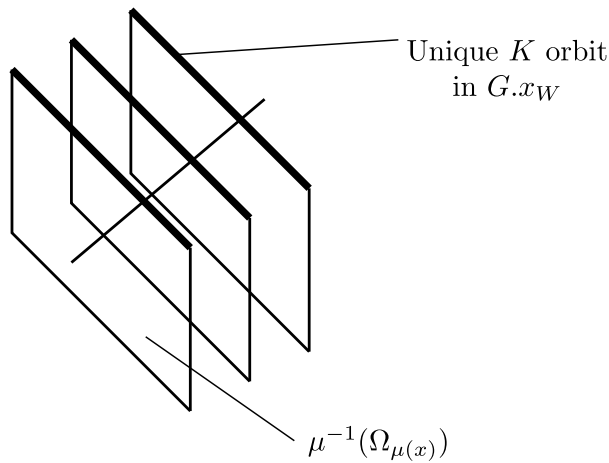}

\caption{The structure of $\mu^{-1}(\Omega_{\alpha})$, where $\alpha$ is
in the interior of the polytope $\Psi(X_W)$.}
\label{figure:slices}
\end{figure}

\subsection{The boundary of the Kirwan polytope}
\label{sec:boundary}

In our analysis at the beginning of this section we did only consider quantum
states that are mapped into the interior of the Kirwan polytope. We will now
prove that any pure quantum state of three qubits that is mapped to the
boundary of the Kirwan polytope $\Psi(\mathbb P(\mathcal H))$ is, up to local
unitaries, uniquely determined by the spectra of its one-body reduced density
matrices.

Let us therefore consider a quantum state $x \in \mathbb P(\mathcal H)$ that
is mapped to the boundary of the Kirwan polytope and write
\begin{equation}
\mu(x) = (\diag(-i\lambda_1,i\lambda_1), \diag(-i\lambda_2,i\lambda_2),
\diag(-i\lambda_3,i\lambda_3)),
\end{equation}
according to (\ref{positive weyl chamber}). Then the minimal eigenvalue $p_j$
of the reduced density matrix $\rho_j$ is given by $p_j = \frac 1 2 -
\lambda_j$. We distinguish two cases:

\paragraph{Non-degenerate case}
We shall first treat the case where the reduced density matrices of $x$ all
have non-degenerate eigenvalue spectrum; that is, $\mu(x)$ is contained in
the interior of the positive Weyl chamber. We may assume by symmetry that
$\mu(x)$ is contained in the face corresponding to the equation $p_1 = p_2 +
p_3$ (cf.~Fact \ref{higuchi}), i.e.
\begin{equation}
  \label{interior face}
  - \lambda_1 + \lambda_2 + \lambda_3 = \frac 1 2.
\end{equation}
That is, $\mu(x)$ is orthogonal to (annihilated by) the Lie algebra element
$i \xi \in \mathfrak k$, where
\begin{equation*}
  \xi = \frac 1 2 \left( \diag(1, -1), \diag(1, -1), \diag(1, -1) \right),
\end{equation*}
and Fact \ref{range_mu} implies that $\mathbb C \xi$ is contained in the Lie
algebra of the isotropy subgroup of $x$. It follows that $x = [v]$ for some
eigenvector $v \in \mathcal H$ of $\xi$ with eigenvalue $\nicefrac 1 2$. By
diagonalizing the action of $\xi$ on $\mathcal H$ we find that
\begin{equation*}
  x = [c_1 \ket{000} + c_2 \ket{110} + c_3 \ket{101}]
\end{equation*}
for some constants $c_j$. Any such state is contained in $X_W$, the closure
of the SLOCC class of the $W$-state, as can be seen by applying $\ket 0
\leftrightarrow \ket 1$ to the first subsystem and comparing with (\ref{some
W class states}). We can therefore use Theorem \ref{W sep} to conclude that
$x$ is determined up to local unitaries by the eigenvalues of its reduced
density matrices.

\paragraph{Degenerate case}
We will now treat the case where at least one of the reduced density matrices
is maximally mixed. Without loss of generality, we may assume that $\mu(x)$
is contained in the face of the Kirwan polytope defined by $p_1 = \frac 1 2$,
i.e.\ $\lambda_1 = 0$. Note that we \emph{cannot} apply the same reasoning as
above, since the faces $\lambda_j = 0$ arise from intersecting $\mu(\mathbb
P(\mathcal H))$ with the positive Weyl chamber and not from the geometry of
the momentum map. By the Schmidt decomposition, and up a local unitary on the
first subsystem,
\begin{equation*}
  x = [\frac 1 {\sqrt 2} \left( \ket 0 \otimes \ket \phi +
\ket 1 \otimes \ket \psi \right)]
\end{equation*}
for orthogonal vectors $\langle \phi | \psi \rangle = 0$. The two-qubit
reduced density matrix $\rho_{23}$ of any such state is a normalized
projector onto the two-dimensional subspace of $\mathbb C^2 \otimes \mathbb
C^2$ with basis vectors $\ket \phi$ and $\ket \psi$; conversely, any other
choice of orthonormal basis gives rise to a local unitarily equivalent state.

Let us denote by $\mathbb G(2,4)$ the Grassmannian consisting of
two-dimensional subspaces $\mathcal K \subseteq \mathbb C^2 \otimes \mathbb
C^2$. Similarly to the case of the projective space, we can consider the
tensor product action of $G' = SL(2) \times SL(2)$ and of its maximal compact
subgroup $K' = SU(2) \times SU(2)$. A momentum map for the $K'$-action is
given by
\begin{equation*}
  \mu' : \mathbb G(2, 4) \ni \mathcal K \mapsto i \left(\rho_2 -
  \frac 1 2 I, \rho_3 - \frac 1 2 I\right) \in \mathfrak{k}'^\ast, \quad
  \mathfrak{k}^\prime=\mathfrak{su}(2)\oplus\mathfrak{su}(2),
\end{equation*}
where $\rho_2$ and $\rho_3$ denote the reduced density matrices of $\rho_{23}
= \frac 1 2 P_{\mathcal K}$, the normalized projector onto the subspace
$\mathcal K$. In view of Fact \ref{Brion}, it suffices to establish
sphericality of this Grassmannian with respect to the action of $G'$:

\begin{prop}
  The $G'$-variety $\mathbb G(2,4)$ is spherical.
\end{prop}
\begin{proof}
  Let $B'$ denote the Borel subgroup consisting of lower triangular matrices
  in $G'$. As in the proof of Proposition \ref{borel}, we will show that
  there exists a point $\mathcal K \in \mathbb G(2,4)$ at which
  $\dim_{\mathbb C} T_\mathcal K (B' . \mathcal K) = \dim_{\mathbb C} \mathbb
  G(2,4) = 4$ (noting that $\mathbb G(2,4)$ is smooth). It will be convenient
  to work with coordinates. Let us therefore consider the Pl\"ucker embedding
  \begin{equation*}
    \mathbb G(2, 4) \rightarrow \mathbb P(\bigwedge^2 \mathbb C^4),
    \quad
    \Span_{\mathbb C} \{ \ket \phi, \ket \psi \} \mapsto [\ket \phi \wedge \ket \psi].
  \end{equation*}
  The image of the subspace $\mathcal K = \Span_{\mathbb C} \{ \ket{++}, \ket{--} \}$,
  where $\ket{\pm} = \frac 1 {\sqrt 2} ( \ket 0 \pm \ket 1 )$, is
  \begin{equation*}
    x = [\ket{++} \wedge \ket{--}] = [\ket{00} \wedge \ket{01} +
\ket{00} \wedge \ket{10} - \ket{01} \wedge \ket{11} -
\ket{10} \wedge \ket{11}].
  \end{equation*}
It follows that the tangent space $T_{\mathcal K} (B' . \mathcal K)$ is
spanned by the vectors
  \begin{eqnarray*}
    (E_{21} \otimes I) x \propto \ket{00} \wedge \ket{11} - \ket{01} \wedge \ket{10}, \\
    (I \otimes E_{21}) x \propto \ket{00} \wedge \ket{11} + \ket{01} \wedge \ket{10}, \\
    ((E_{11} - E_{22}) \otimes I) x \propto \ket{00} \wedge \ket{01} + \ket{10} \wedge \ket{11}, \\
    (I \otimes (E_{11} - E_{22})) x \propto \ket{00} \wedge \ket{10} + \ket{01} \wedge \ket{11},
  \end{eqnarray*}
which are orthogonal to $x$ and linearly independent over $\mathbb C$. We
conclude that the tangent space is of complex dimension four.
\end{proof}

\noindent
In summary, we have proved the following result:

\begin{thm}
  Let $x$ be a pure quantum state of three qubits such that $\mu(x)$ is
  contained in the boundary of the Kirwan polytope $\Psi(\mathbb P(\mathcal
  H))$ (i.e.\ its eigenvalues satisfy at least one of the inequalities in Fact
  \ref{higuchi} with equality). Then $x$ is up to local unitaries uniquely
  determined by $\mu(x)$, i.e.\ by the spectra of its one-body reduced density
  matrices.
\end{thm}

\section{Summary}

In order to determine which pure states of three qubits are up to local
unitaries uniquely determined by the spectra of their reduced density
matrices, we have analyzed the change of the structure of the fiber
$\Psi^{-1}(\alpha)$ as $\alpha$ varies in the Kirwan polytope
$\Psi(\mathbb{P}(\mathcal{H}))$. We have shown that $M_\alpha =
\Psi^{-1}(\alpha)/K$ is generically a two-dimensional space. For the points
in the boundary of the Kirwan polytope, the situation is rather different. We
have showed that in this case $\Psi^{-1}(\alpha)/K$ is a single point, i.e.\
the dimension drops by two compared with the points inside the interior. We
have therefore identified all $K$-orbits that are uniquely determined by the
spectra of the reduced density matrices.
In addition, we have examined the Kirwan polytopes $\Psi(\overline{G.x_j})$
for all six three-qubit SLOCC classes (i.e.\ for the closures of the orbits of
the complexified group $G=K^{\mathbb{C}}$) and their mutual relation. In
particular, we have proved that states from the so-called $W$ SLOCC class are
up to local unitaries separated by $\Psi$, i.e.\ each $K$-orbit inside
$\overline{G.x_W}$ is characterized by the collection of spectra of the
one-qubit reduced density matrices. This statement generalizes in a
straightforward way to any number of qubits.

Interestingly, the drop of the dimension of $M_\alpha$ on the boundary of the
Kirwan polytope has the following counterpart in
\cite{ChristandlDoranKousidisWalter}: The probability density $f(\alpha)$ of
the eigenvalue distribution of the reduced density matrices of a randomly
chosen pure state of three qubits vanishes precisely on the boundary of the
Kirwan polytope. Since it is well-known that $f(\alpha) = \vol M_\alpha$ for
regular points of the momentum map \cite{duistermaatheckman82,meinrenken97},
it is reasonable to wonder whether an analogous statement could hold more
generally (the main result of \cite{meinrenkensjamaar99} suggests that this
might in fact be true).

We also noticed that the polytope $\Psi(X_W)$ associated with the $W$ SLOCC
class is of the same dimension as the polytope $\Psi(X_{\GHZ})$ corresponding
to the $\GHZ$ SLOCC class, whereas for bi-separable and separable states the
Kirwan polytope is of strictly lower dimension. On the other hand, it is
known that for three qubits only the $W$ and $\GHZ$ SLOCC classes represent
genuinely entangled states. This intriguing relationship suggests that by
looking at the polytopes corresponding to different SLOCC classes one can
decide to what extent states from this class are entangled. We believe that
this is not a coincidence and conjecture that similar phenomena should be
present for $N$-qubit systems, where $N > 3$.

The symplectic methods of the present paper can also be applied and extended to
other problems of quantum entanglement theory. In \cite{SOK12}, two of the authors have analyzed from
the to\-po\-lo\-gi\-cal perspective all SLOCC classes of pure states for both distinguishable
and indistinguishable particles using geometric invariant theory and momentum map geometry, resulting in a division of all SLOCC classes into physically meaningful groups of families. Based on the above results, they presented in \cite{SOK12a} a general algorithm, working for arbitrary systems of distinguishable and indistinguishable particles, for finding the above mentioned groups of families. The algorithm provides a certain, physically meaningful, parametrization of SLOCC classes by critical sets of the so-called total variance function of a state.
Independently, in \cite{WalterDoranGrossChristandl} one of the authors has in collaboration with others combined symplectic geometry and geometric invariant theory in a novel approach to the study of multiparticle entanglement, resulting in a systematic way of classifying multiparticle entanglement that can be witnessed efficiently and robustly in experiments.

\section*{Acknowledgments}
We gratefully acknowledge the support of SFB/TR12 Symmetries and Universality
in Mesoscopic Systems program of the Deutsche Forschungsgemeischaft, a grant
of the Polish National Science Center under the contract number
DEC-2011/01/M/ST2/00379, Polish MNiSW grant no.\ IP2011048471, the Swiss
National Science Foundation (grant PP00P2--128455), the German Science
Foundation (grant CH 843/2--1), and the National Center of Competence in
Research `Quantum Science and Technology'. The authors are especially in debt
to Alan Huckleberry for many valuable discussions concerning symplectic and
algebraic geometry and for his constant interest in their work. The second
author would also like to thank Matthias Christandl, Brent Doran, David
Gross, and Eckhard Meinrenken for many fruitful discussions.

\section*{References}

\end{document}